\documentclass[12pt]{amsart}
\UseRawInputEncoding
\usepackage{amsmath,amsthm,txfonts,bbm,graphicx,verbatim}

\usepackage[
  top=1in,
  bottom=1in,
  left=1.15in,
  right=1.15in,
  heightrounded,
]{geometry}

\theoremstyle{plain}
\newtheorem{theorem}{Theorem}%[section]

\newtheorem{proposition}[theorem]{Proposition}

\newtheorem{remark}[theorem]{Remark}
\newtheorem{rhp}[theorem]{Riemann--Hilbert Problem}

\begin{document}

\title[Algebro-Geometric Finite Gap Primitive Solutions]{Algebro-Geometric Finite Gap Solutions to the Korteweg--de Vries Equation as Primitive Solutions}
\author{Patrik V. Nabelek}
\address{Department of Mathematics, Kidder Hall 368, Oregon State University, Corvallis, OR 97331-4605}
\email{nabelekp@oregonstate.edu}

\maketitle

\begin{abstract}
In this paper we show that all algebro-geometric finite gap solutions to the Korteweg--de Vries equation can be realized as a limit of N-soliton solutions as N diverges to infinity (see remark \ref{rmk:shifty} for the precise meaning of this statement).
This is done using the the primitive solution framework initiated by \cite{DZZ16,ZDZ16,ZZD16}.
One implication of this result is that the N-soliton solutions can approximate any bounded periodic solution to the Korteweg--de Vries equation arbitrarily well in the limit as N diverges to infinity.
We also study primitive solutions numerically that have the same spectral properties as the algebro-geometric finite gap solutions but are not algebro-geometric solutions.
\end{abstract}

\section{Introduction}

The Korteweg--de Vries (KdV) equation
\begin{equation} \label{eq:KdV} u_t-6uu_x+u_{xxx} = 0 \end{equation}
was originally derived to describe soliton waves on the surface of a channel \cite{KdV}.
In particular, the KdV equation is a weakly nonlinear description of channeled surface waves.
The KdV equation has also seen use in oceanography \cite{Osborne2010}, including as a weakly nonlinear description of nonlinear internal waves \cite{HM06}.
The KdV equation is also important because it is the prototypical example of an infinite dimensional, completely integrable Hamiltonian system \cite{ZF71}.
The differential equation form of the infinitely many conservation laws making the KdV completely integrable are known as the higher KdV equations.
A solution $u(x,t)$ to the KdV equation corresponds to an isospectral evolution of 1D Schr\"{o}dinger operators $-\partial_x^2+u(x,t)$ obeying the Lax equation \cite{Lax} (this is also true for solutions $u$ to the higher KdV equations).
The solution $u(x,t)$ is also called a potential, because for each fixed $t$ the function $u(x,t)$ can be interpretated as a potential energy when it appears in the Schr\"{o}dinger operator.

Two important boundary conditions for which exact formulas for the solutions of the KdV equations are known are localized solutions and solutions that are periodic in $x$.
The localized case was solved by the inverse scattering transform, and an interesting class of solutions are the N-soliton solutions, which occur when the Schr\"{o}dinger operator is reflectionless \cite{GGKM,ZMNP84}.
The periodic case was solved using algro-geometric methods instigated by Novikov \cite{No74}, Marchenko \cite{Mar74} and Lax \cite{Lax75}, and resolved independently by Dubrovin \cite{Dubrovin75,Dubrovin76}, Matveev--Its \cite{ItsMatv75a,ItsMatv75b,ItsMatv76,Mat75}, and McKean--van Morebeke \cite{MckMo75} in the important case of algebro-geometric finite gap solutions to the KdV equation (also see \cite{DN74,DMN76}).
Marchenko--Ostrovskii showed that the space of periodic finite gap solution is dense in the space of periodic solutions \cite{MaOs75}, and McKean--Trubowitz showed that the algebro-geometric methods extend to the general smooth periodic case \cite{MckTr76,MckTr78}.
An effective way of computing periodic finite gap solutions to the KdV equation that approximate smooth periodic solutions to the KdV equation arbitrarily well is to use the isoperiodic flows introduced by Grinevich--Schmidt to close all but a finite number of spectral gaps because the widths of the gaps decay exponentially at high energy \cite{GS95}. 
It should be noted that the algebro-geometric methods allow computation of interesting multi/quasi-periodic solutions that are outside the space of periodic solutions.
A Riemann--Hilbert problem approach to infinite gap periodic and quasi-periodic solutions to the KdV equation was discussed in \cite{MN19}.

Primitive solutions to the KdV equation were introduced in \cite{DZZ16,ZDZ16,ZZD16} as elements of the closure of the N-soliton solutions to the KdV equation with respect to the topology of uniform convergence in compact sets.
This space was originally considered by Marchenko \cite{Mar91}.
Primitive solutions were derived using the dressing method introduced by Zakharov--Manakov \cite{ZakharovManakov}.
The primitive solution method was also adapted to the Kaup--Broer system in \cite{NZ20}.
The potentials are determined by a pair of functions $R_1$ and $R_2$ called dressing functions, and the potentials can be computed via solving a system of singular integral equations.
%A special case of symmetric primitive potentials was discussed in \cite{PZZ19}.  
An intriguing aspect of this initial work was that it allows computation of potentials that have the same spectra as the algebro-geometric finite gap potentials, but are not the usual algebro-geometric finite gap potentials.
These potentials can have either simple or doubly degenerate continuous spectra on the interiors of their spectral bands.
The algebro-geometric finite gap potentials can only have doubly degenerate spectra.
But even when the primitive potential has only a doubly degenerate continuous spectrum on the interiors of its spectral bands, the primitive potential need not have any periodicity or quasi-periodicity properties.
In \cite{DZZ16,ZDZ16,ZZD16,NZZ18} dressing functions for cnoidal waves were determined.
The main results of this paper were reviewed in \cite{DNZZ20}, but the details of this result are provided in this paper.

\begin{remark} \label{rmk:shifty}
To make the precise meaning of statements that follow, it is convenient to introduce some terminology.
We first note that if $u(x,t)$ solves the KdV equation then so does
\begin{equation} \label{eq:shifty} u(x-6 C t,t) + C\end{equation}
for any constant $C$.
We then use the terminology ``shifted N-soliton solution'' and ``shifted primitive solution'' to refer to solutions to the KdV equation that can be produced from the N-soliton solution and primitive solutions via the symmetry \eqref{eq:shifty}.
\end{remark}
 
In this paper we show that shifted N-soliton solutions are dense in the finite gap solutions with respect to the topology of uniform convergence in compact sets by computing dressing functions corresponding to the finite gap potentials.
The fact that the shifted N-soliton solutions are dense in the space of bounded periodic solutions to the KdV equation then follows from the fact that periodic finite gap potentials are a dense subset of the bounded periodic potentials \cite{MaOs75}.
The density of the N-soliton solutions in the algebro-geometric finite gap solutions and periodic solutions was proved by Marchenko \cite{Mar91}.
The primitive potential construction gives an effective method of taking this closure and producing explicit sequences of shifted N-soliton solutions converging to the algebro-geometric finite gap solutions.
The reverse limit of computing N-soliton solutions as algebro-geometric finite gap solutions in the limit in which the genus $N$ hyper-elliptic spectral curve of the N-gap solution degenerates to a rational curve is well known, and was completed by Novikov, Matveev, Its and Dubrovin \cite{DMN76,ItsMatv76,ZMNP84}.

The existence of algebro-geometric finite gap solutions as limits of shifted N-soliton solutions is also of interest because it shows the possibility of approximating hyper elliptic functions by rational functions away from the branch points/poles.
In essence, the meromorphic function theory on the hyper elliptic curve corresponding to the hyper elliptic function is approximated by holomorphic functions on rational curves with $N$ degenerate points.
This point of view is discussed in more detail in \cite{NZZ18}.

This paper is structured as follows.
\begin{itemize}
\item In section 2 we review the primitive potential theory, and state theorem \ref{thm:fingap}, which is the main result of this paper.
\item In section 3 we construct the KdV equation spectral curves in a manner that naturally produces coordinates that are useful in deriving the primitive potential integral equation.
\item In section 4 we construct the dressing functions for the algebro-geometric finite gap solutions directly from the Baker--Akhiezer function on the spectral curve.
\item In section 5 we numerically compute algebro-geometric finite gap solutions to the KdV equation via the primitive solution system of singular integral equations for small genus.
We also numerically compute primitive solutions that have the same spectral properties of the finite gap solutions, but are not finite gap solutions.
\item In section 6 we provide some concluding remarks.
\end{itemize}

\section{Finite Gap Solutions as Primitive Potentials}

Let $R_1$ and $R_2$ be nonnegative H\"{o}lder continuous real functions.
Primitive solutions to the KdV equation are defined in terms of the solution $f,g$ to the system of singular integral equations
\begin{align}
f(p,x,t)+\frac{R_1(p)}{\pi}e^{-2px+8p^3t}\left[\int_{k_1}^{k_2} \frac{f(q,x,t)}{p+q}dq+\fint_{k_1}^{k_2}\frac{g(q,x,t)}{p-q}dq\right]=R_1(p)e^{-2px+8p^3x}, \label{eq:f} \\
g(p,x,t)+\frac{R_2(p)}{\pi}e^{2px-8p^3t}\left[\fint_{k_1}^{k_2} \frac{f(q,x,t)}{p-q}dq+\int_{k_1}^{k_2}\frac{g(q,x,t)}{p+q}dq\right]=-R_2(p)e^{2px-8p^3t} \label{eq:g}
\end{align}
for each fixed value of $x,t$.
A $t$ dependent family $u(x,t)$ of primitive potentials can be constructed as
\begin{equation} u(x,t) = \frac{2}{\pi} \frac{\partial}{\partial x} \int_{k_1}^{k_2} f(q,x,t) + g(q,x,t) dq, \label{eq:u} \end{equation}
such that
\begin{equation} \psi(k,x,t) = e^{-i k x-4ik^3t} \left( 1 + \frac{i}{\pi} \int_{k_1}^{k_2} \frac{f(q,x,t)}{k - iq} dq + \frac{i}{\pi} \int_{k_1}^{k_2} \frac{g(q,x,t)}{k + iq} dq \right) \end{equation}
solves the Schr\"{o}dinger equation
\begin{equation} -\psi_{xx}(k,x,t) + u(x,t) \psi(k,x,t) = k^2 \psi(k,x,t) \end{equation}
for complex $k^2 \notin \sigma(-\partial_x^2+u(x,t)) \cap [-k_2^2,-k_1^2]$, and $u(x,t)$ solves the KdV equation (\ref{eq:KdV}).

The solutions $\psi(k,x,t)$ with $k \in \mathbb{R}$ are physical solutions of the Schr\"{o}dinger equation with energy $k^2$ (by physical solutions we mean bounded solutions).
The boundary values of $\psi_\pm(k,x,t)$ for $k^2 \in \sigma(-\partial_x^2 + u(x)) \cap [-k_2^2,-k_1^2]$ when $k^2$ is a doubly degenerate point in the continuous spectrum allows computation of the other physical solutions with either $\psi_\pm(ip,x,t)$ or $\psi_\pm(-ip,x,t)$ giving a basis of 2 linearly independent improper eigenfunctions (on endpoints of the spectrum, the physical solution can be determined by the singular behavior of $\psi$).
Alternatively, we may take $\varphi^+(p,x,t) = e^{px-4p^3t} f(p,x,t)$ and $\varphi^-(p,x,t) = e^{-px+4p^3t}g(p,x,t)$ as a basis of physical solutions with energy $-p^2$ as discussed in \cite{DZZ16,ZDZ16,ZZD16}.
In the case where either $R_1(p) = 0$ or $R_2(p) = 0$ for $-p^2 \in [-k_2^2,-k_1^2] \cap \sigma(-\partial_x^2+u(x,t))$ the continuous spectrum is simple and the nonzero choice of $\varphi(p,x,t) = e^{px-4p^3t} f(p,x,t)$ or $\varphi(p,x,t) = e^{-px+4p^3t}g(p,x,t)$ gives the single physical solution.

One can use uniform grids with $N/2$ points to discretize the $p$ dependance of $f$ and $g$ with the uniform grid for $g$ staggered relative to the uniform grid for $f$ and approximate the integrals in (\ref{eq:f},\ref{eq:g}) using Riemann sums; in this case the system of singular integral equations becomes a finite dimensional linear system and the dressing method gives an exact N-soliton solution to the KdV equation \cite{DZZ16,ZDZ16,ZZD16}.
Therefore, we can approximate a primitive solution to the KdV equation as accurately as we desire in a compact subset of the $x,t$-plane by an N-soliton solution to the KdV equation in the above manner by making $N$ sufficiently large.

In previous papers on primitive solutions to the KdV equation, $R_1$ and $R_2$ were assumed to be nonnegative H\"{o}lder continuous functions on $[k_1,k_2]$ \cite{DZZ16,ZDZ16,ZZD16,NZZ18}.
However, we must now weaken this to the assumption that $R_1$ and $R_2$ are nonnegative and H\"{o}lder continuous on their supports.
This allows $R_1$ and $R_2$ to have jump discontinuities between positive numbers and $0$ on the interval $[k_1,k_2]$.

\begin{theorem} \label{thm:fingap}
Consider an increasing sequence $\{\kappa_j\}_{j = 1}^g$ with
\begin{equation} \label{eq:kcondition} 0 < k_1 < \kappa_1 < \kappa_2 < \dots < \kappa_{2g} < k_2 < \infty. \end{equation}
Let $u(x,t)$ be an algebro-geometric finite gap solution to the KdV equation such that at each fixed time $t$ the potential $u(x,t)$ is an algebro-geometric finite gap potential with spectrum
\begin{equation} \label{eq:spec} \sigma(-\partial_x^2 + u(x,t)) = [-\kappa_{2g}^2,-\kappa_{2g-1}^2] \cup \dots \cup [-\kappa_4^2, -\kappa_3^2] \cup [-\kappa_2^2, -\kappa_1^2] \cup [0,\infty).  \end{equation}
Then $u(x,t)$ is the primitive solution determined by
\begin{align} \label{eq:agpsR1} & R_1(p) = \exp\left( \sum_{j=1}^g a_j p^{2j-1} \right) \sum_{\ell=1}^g \mathbbm{1}_{[\kappa_{2\ell-1}, \kappa_{2 \ell}]} (p),  \\
\label{eq:agpsR2}& R_2(p) = \exp\left( -\sum_{j=1}^g a_j p^{2j-1} \right) \sum_{\ell=1}^g \mathbbm{1}_{[\kappa_{2\ell-1}, \kappa_{2 \ell}]} (p),   \end{align}
where $a_j$ are real constants, and $\mathbbm{1}_{[\kappa_{2 \ell-1},\kappa_{2 \ell}]}$ is the indicator function of $[\kappa_{2 \ell-1},\kappa_{2 \ell}]$.
Conversely, if $u(x,t)$ is a primitive solution to the KdV equation determined by dressing functions $R_1$ and $R_2$ given by \eqref{eq:agpsR1} and \eqref{eq:agpsR2} for some choice of real $\{a_j\}_{j = 1}^g$ and $\{\kappa_j\}_{j=1}^{2g}$ satisfying \eqref{eq:kcondition}, then $u(x,t)$ is an algebro-geometric finite gap solution with spectrum of the form \eqref{eq:spec}.
\end{theorem}

All of the choices of dressing functions $R_1$ and $R_2$ of the form (\ref{eq:agpsR1},\ref{eq:agpsR2}) produce finite gap solutions. These are only finite gap solutions because there are only terms involving powers of $j$ for $j = 1,3,\dots 2g-3,2g-1$ appearing in the exponents. This result is due to the invertibility of the matrix $\Omega$ appearing in (\ref{eq:auxcond}). If any additional terms were added for example, the matrix $\Omega$ would not be invertible. 
Most other choices will likely lead to a solution that is not finite gap, but is likely asymptotically equivalent to a finite gap solutions.
In the one gap case, this hypothesis is supported by the rigorous analysis in \cite{GGJM18}.

\begin{remark}
Theorem \ref{thm:fingap} implies that:
\begin{enumerate}
\item Any algebro-geometric finite gap potential can be realized as a shifted primitive solution (recall that a shifted primitive solution was defined in remark \ref{rmk:shifty}).
\item Any shifted primitive solution determined by dressing functions of the form \eqref{eq:agpsR1} and \eqref{eq:agpsR2} is an algebro-geometric finite gap potential.
\end{enumerate}
\end{remark}

Theorem \ref{thm:fingap} is proven by starting with a Baker--Akheizer function expressed in a convenient coordinate system on the spectral curve, and then deriving a system of integral equations of the form \eqref{eq:f}, \eqref{eq:g}.
All potentials formed by the choice of dressing function from theorem \ref{thm:fingap} are related to the solution determined by
\begin{align} \label{eq:constdress1} & R_1(p) =  \sum_{\ell=1}^g \mathbbm{1}_{[\kappa_{2\ell-1}, \kappa_{2 \ell}]} (p),  \\
\label{eq:constdress2} & R_2(p) = \sum_{\ell=1}^g \mathbbm{1}_{[\kappa_{2\ell-1}, \kappa_{2 \ell}]} (p), \end{align}
by a transformation that maintains the period of any periodic primitive potential.

Space and time translations are sufficient to transform any two-gap potential into the primitive potential corresponding to dressing functions of the form (\ref{eq:constdress1},\ref{eq:constdress2}).
The primitive potentials corresponding to dressing functions (\ref{eq:constdress1},\ref{eq:constdress2}) are examples of symmetric primitive potentials and primitive solutions discussed in \cite{NZZ18}.
The method for computing the Taylor coefficients of these solutions about $(x,t) = (0,0)$ in the one-gap case discussed in \cite{NZZ18} can be easily adapted to apply to the primitive solutions to the KdV equation determined by the dressing functions of the form (\ref{eq:constdress1},\ref{eq:constdress2}).

%\begin{corollary}
%The space of real periodic potentials in $C^\infty(\mathbb{R})$ are in the closure in the topology of uniform convergence in compact sets of the space of potentials of the form $u(x) = u_N (x) + c$ where
%\begin{equation} u_N(x) = -2 \frac{\partial^2}{\partial x^2} \log \left( \sum_{I\subset \{1,\ldots,N\}}\prod_{(i,j)\subset I,\,i<j} \frac{(\tilde \kappa_i- \tilde \kappa_j)^2}{(\tilde \kappa_i+ \tilde\kappa_j)^2}\prod_{i\in I}q_i e^{-2\tilde \kappa_i x}\right),\quad q_i>0  \end{equation}
%is a Bargmann potentials with $N$ eigenvalues and $c$ is an arbitrary constant.
%\end{corollary}

\section{Baker--Akheizer Function on the ``k-Plane''}

Consider the curve $\tilde \Sigma$ defined by
\begin{equation} w^2 = \prod_{j = 1}^{2g} (k^2 + \kappa_j^2) = P_{4g}(k). \end{equation}
Also, consider the involution $\iota(k,w) = (-k,-w)$.
The genus of $\tilde \Sigma$ is $g' = 2g-1$. 
The spectral curve for the KdV equation can then be formed as the quotient space $\Sigma = \tilde \Sigma/\left<\iota\right>$.
The surface $\tilde \Sigma$ is the double cover of $\Sigma$.
The curve $\Sigma$ is a Riemann surface with genus $g$ and is homeomorphic to the curved defined by
\begin{equation} u^2 = v \prod_{j = 1}^{2g} (v + \kappa_j^2),  \end{equation}
which is the standard representation of the KdV spectral curve.
However, the coordinates $k$ produced from a single sheet of the double cover are more natural when comparing to the primitive solutions.
If $A \subset \mathbb{C}$ we will use the notation $iA = \{ik : k \in A\}$.
Let us put the branch cuts of $\sqrt{P_{4g}(k)}$ on $i \Gamma$ where
\begin{equation} \Gamma = \bigcup_{j = 1}^g [-\kappa_{2j},-\kappa_{2j-1}] \cup [\kappa_{2j-1},\kappa_{2j}] \end{equation}
is oriented from left to right (so $i \Gamma$ is oriented from down to up).
We will abuse notation, and use $k$ to represent both the complex number $k \in \mathbb{C} \setminus i \Gamma$ and the corresponding point on $\Sigma$.
We will use the notation $\left< i \kappa_j \right>$ to indicate the Wierstrass point on $\Sigma$ as $k$ approaches $\pm i \kappa_j$.

The Abelian differentials of the first kind on $\Sigma$ can be produced by computing explicitly the above homeomorphism, and then computing a pull back.
However, we can also compute them using the above construction by computing the holomorphic differentials on $\tilde \Sigma$ and then finding those that are invariant under $\iota$.
A holomorphic differential $\omega$ on $\tilde \Sigma$ is uniquely expressed as
\begin{equation} \eta = \sum_{n = 1}^{2g-1} c_n \frac{k^{n-1}}{\sqrt{P_{4g}(k)}} dk, \end{equation}
and so $\iota$ acts as
\begin{equation} \iota^* \eta = \sum_{n = 1}^{2g-1} c_n \frac{(-1)^{n-1} k^{n-1}}{\sqrt{P_{4g}(k)}} dk. \end{equation}
Therefore, $\eta$ is invariant if and only if $c_{2j} = 0$ for $j = 1,2,\dots,g-1$.
This means that a basis of Abelian differentials of the first kind on $\Sigma$ is
\begin{equation} \eta_j = \frac{k^{2j-2}dk}{\sqrt{P_{4g}(k)}} \end{equation}
for $j = 1,2,\dots,g$.
The basis of Abelian differentials of the first kind on $\Sigma$ is $g$ dimensional so $\Sigma$ has genus $g$.
Let $\boldsymbol{\eta}$ be the $g$ dimensional vector of differentials with entries $\eta_j$.

We now introduce a canonical homology basis $\{a_j,b_j\}_{j = 1}^g$ for $H_1(\Sigma)$ satisfying $a_i \circ b_j = \delta_{ij}$, $a_i\circ a_j = 0$ and $b_i \circ b_j = 0$, where $\circ$ indicates the minimal intersection number for homology elements.
With the homology basis in hand, we can compute the period matrix $M$ with entries
\begin{equation} M_{ij} = \int_{a_j} \eta_i. \end{equation}
Then the entries $\omega_j$ of
\begin{equation} \boldsymbol{\omega} =2 \pi i M^{-1} \boldsymbol{\eta}\end{equation}
form a basis of Abelian differentials of the first kind on $\Sigma$ normalized by
\begin{equation} \int_{a_j} \omega_i = 2 \pi i \delta_{ij}. \label{eq:normad} \end{equation}
This basis of Abelian differentials of the first kind and the choice of base point at $\infty$ allows us to compute:
\begin{enumerate}
\item The Abel map $\mathbf{A}(k)$ with entries
\begin{equation} \mathbf{A}(k) = \int_\infty^k \boldsymbol{\omega}\end{equation}
mapping $\Sigma$ into the Jacobi variety. The discontinuities on the $k$-coordinate expression of the Abel map $\mathbf{A}(k)$ for $k \in i \Gamma$ correspond to $g$ disjoint circles in the Jacobi variety.
\item The Abel map on degree $g$ divisors $\boldsymbol{\delta}$ is
\begin{equation} \mathbf{A}(\boldsymbol{\delta}) = \sum_{j = 1}^g \mathbf{A}(\delta_j).\end{equation}
\item The Riemann matrix $B$ with negative definite real part
\begin{equation} B_{ij} = \int_{b_j} \omega_i. \end{equation}
\item The vector of Riemann constants $\mathbf{K}$ with entries
\begin{equation} K_j =  \frac{2 \pi i + B_{jj}}{2} - \frac{1}{2 \pi i} \sum_{\ell \ne j} \int_{a_{\ell}} A_j(k) \omega_{\ell}. \end{equation}
\end{enumerate}

We define a coordinate $\zeta = z^{-1}$ for $\Sigma$ at $k=\infty$.
Let $\omega^{(n)}$ be Abelian differentials of the second kind on $\Sigma$ with poles at $\infty$ with principle parts
\begin{equation} \omega^{(n)} \sim dk^n = d \zeta^{-n} = -n\zeta^{-n-1} d\zeta  \end{equation}
and
\begin{equation} \int_{a_j} \omega^{(n)} = 0. \end{equation}
This Abelian differential has the form
\begin{equation}  \omega^{(n)} = \frac{n k^{2g-1+n}}{\sqrt{P_{4g}(k)}} + \sum_{j =1}^g c_j \omega_j. \end{equation}
An important aspect of $\omega^{(n)}$ is the vector $\boldsymbol{\Omega}^{(n)}$ with entries
\begin{equation} \Omega_j^{(n)} = \int_{b_j} \omega^{(n)}. \end{equation}

The Riemann theta function $\theta: \mathbb{C}^g \to \mathbb{C}$ is defined by
\begin{equation} \theta(\mathbf{z},B) = \sum_{\mathbf{n} \in \mathbb{Z}^g} \exp \left( {\frac{1}{2} \mathbf{n} \cdot B \mathbf{n} + \mathbf{n} \cdot \mathbf{z}} \right). \end{equation}
The Riemann theta function converges uniformly because $B$ has negative definite real part.

The sets $O_1 = i(-\kappa_1,\kappa_1) \cup \left<i\kappa_1\right>$ and
\begin{equation} O_j = i (-\kappa_{2j-1},-\kappa_{2j-2}) \cup i(\kappa_{2j-2},\kappa_{2j-1}) \cup \left<i \kappa_{2j-2}\right> \cup \left<i \kappa_{2j-1}\right> \end{equation}
for $j=2,\dots,g$ form a collection of $g$ real ovals of $\Sigma$.
Consider a degree $g$ divisor $\boldsymbol{\delta} \in \Sigma^g$ consisting of the direct sums of points $\delta_j \in O_j$.
When $\Sigma$ is the spectral curve for a periodic solution to the KdV equation, then $\boldsymbol{\delta}$ is the Dirichlet divisor of the initial condition.

The Baker--Akhiezer function is the unique function on $\Sigma$ with pole divisor $\boldsymbol{\delta}$ and an asymptotic behavior at $\infty$ of the form
\begin{equation} \psi(k,x,t) = e^{-ikx-4ik^3 t}(1  + O(k^{-1})). \label{eq:BAnorm} \end{equation}
The function $\psi(k,x,t)$ has the explicit formula
\begin{equation} \psi(k,x,t) = \exp\left( -i \int_\infty^k \omega^{(1)} x - 4 i\int_\infty^k \omega^{(3)} t \right) \frac{\theta(\mathbf{A}(k) - \mathbf{A}(\boldsymbol{\delta}) - i \boldsymbol{\Omega}^{(1)} x - 4 i \boldsymbol{\Omega}^{(3)} t - \mathbf{K}, B)\theta( - \mathbf{A}(\boldsymbol{\delta}) - \mathbf{K}, B)}{\theta(\mathbf{A}(k) - \mathbf{A}(\boldsymbol{\delta}) - \mathbf{K}, B)\theta( - \mathbf{A}(\boldsymbol{\delta}) - i \boldsymbol{\Omega}^{(1)} x - 4 i \boldsymbol{\Omega}^{(3)} t  - \mathbf{K}, B)}  \end{equation}
in the $k$-coordinate, and for each $t$ solves the Schr\"{o}dinger equation
\begin{equation} -\psi_{xx}(k,x,t) + u(x,t) \psi(k,x,t) = k^2 \psi(k,x,t) \end{equation}
where $u(k,x,t)$ is given by the Matveev--Its formula \cite{ZMNP84}
\begin{equation} u(x,t) = -2 \frac{\partial^2}{\partial x^2} \theta(-i \boldsymbol{\Omega}^{(1)}x - 4 i \boldsymbol{\Omega}^{(3)} t - \mathbf{A}(\boldsymbol{\delta}) - \mathbf{K}, B). \end{equation}

We make use of the auxiliary Baker--Akhiezer function introduced by Trogdon and Deconinck \cite{TD13} $\psi_{aux}(k)$ with zero divisor $\boldsymbol{\delta}$, pole divisor $\boldsymbol{\gamma}$ consisting of the direct sum of the points $\gamma_j = \left<i \kappa_{2j-1}\right>$, and asymptotic behavior of the form
\begin{equation} \psi_{aux}(k,\boldsymbol{\delta}) = e^{-i \alpha(k,\boldsymbol{\delta})}(1 + O(k^{-1})), \quad \alpha(k,\boldsymbol{\delta}) = \sum_{j = 1}^g t_j (\boldsymbol{\delta}) k^{2j-1}. \end{equation}
Trogdon and Deconinck show that $t_j(\boldsymbol{\delta})$ are real constants determined by solving the linear equation
\begin{equation} \sum_{\ell = 1}^g \Omega_j^{(2\ell-1)} t_\ell (\boldsymbol{\delta}) \equiv A_j(\boldsymbol{\delta}) - A_j(\boldsymbol{\gamma}),  \label{eq:auxcond} \end{equation}
where $\mathbf{A}$ extends to divisors by adding the evaluations of $\mathbf{A}$ on the points in the divisor and $\equiv$ represents equivalence on the Jacobian variety of $\Sigma$; in particular, the matrix $\Omega_j^{(2\ell-1)}$ is invertible \cite{TD13}.
The auxiliary Baker--Akheizer function has the explicit form
\begin{equation} \psi_{aux}(k,\mathbf{t}) = \exp\left( -i \int_\infty^k \sum_{j = 1}^g \omega^{(j)} t_j \right) \frac{\theta\left(\mathbf{A}(k) - \mathbf{A}(\boldsymbol{\delta}) - \mathbf{K} - i \displaystyle{\sum_{j = 1}^g} \boldsymbol{\Omega}^{(j)} t_j , B\right)\theta( - \mathbf{A}(\boldsymbol{\delta}) - \mathbf{K}, B)}{\theta(\mathbf{A}(k) - \mathbf{A}(\boldsymbol{\delta}) - \mathbf{K}, B)\theta\left( - \mathbf{A}(\boldsymbol{\delta}) - \mathbf{K} - i \displaystyle{\sum_{j = 1}^g} \boldsymbol{\Omega}^{(j)} t_j  , B\right)}  \end{equation}
where $\mathbf{t}$ is the $g$ dimesional vector with entries $t_n$.

The intuition behind the auxiliary Baker--Akheizer function is that it tells us how to evolve the initial potential along the higher KdV equation flows until the auxiliary spectral data corresponding to poles of the Baker--Akheizer function lie on $\{\left<i\kappa_{2j-1}\right>\}_{j = 1}^g$.
Alternatively, we can solve the higher Dubrovin equations to evolve the spectral poles of the Baker--Akheizer function.
The right-hand side of equation \eqref{eq:auxcond} differs from the corresponding version in \cite{TD13} by a sign because our asymptotic behavior as $k \to \infty$ used to normalize the Baker--Akheizer functions differs from Trogdon and Deconinck by a sign.

\section{From the Baker--Akheizer Function to the Primitive Potential }

\begin{proposition} \label{prop:RHP}
The function
\begin{equation} \chi(k,x,t) = \xi(k) e^{i\alpha(k, \boldsymbol{\delta})+ikx+4ik^3t} \psi_{aux}(k, \boldsymbol{\delta}) \psi(k,x,t), \label{eq:chi} \end{equation}
where
\begin{equation} \xi(k) = \prod_{j = 1}^g \left( \frac{k^2+\kappa_{2j-1}^2}{k^2+\kappa_{2j}^2} \right)^{\frac{1}{4}}, \end{equation}
solves the following nonlocal scalar Riemann--Hilbert problem:
\end{proposition}
\begin{rhp}
Find a function $\chi(k,x,t)$ such that
\begin{enumerate}
\item $\chi(k,x,t)$ is a holomorphic function for $k \in \mathbb{C} \setminus i \Gamma$.
\item $\chi(k,x,t)$ has continuous non-tangential boundary values
\begin{equation} \chi_+ (ip,x,t) = \lim_{\epsilon \to 0^+} \chi(ip+\epsilon,x,t), \quad \chi_- (ip,x,t) = \lim_{\epsilon \to 0^+} \chi(ip-\epsilon,x,t) \end{equation}
for $k \in i\Gamma \setminus \{ \text{endpoints} \}$.
\item  For fixed $x,t$, there exists a constant $C(x,t)$ such that
\begin{equation} \chi(k,x,t) \le C(x,t) |k \pm i\kappa_j|^{-1/4} \end{equation}
for $k$ in a neighborhood of $\pm i \kappa_j$.
\item The boundary values $\chi_\pm$ satisfy the jump relations
\begin{equation} \chi_+(ip,x,t) = i \text{sgn}(p) e^{2i\alpha(ip,\boldsymbol{\delta})} e^{-2px + 8 p^3 t} \chi_+(-ip,x,t), \end{equation}
\begin{equation} \chi_-(ip,x,t) = -i \text{sgn}(p) e^{2i\alpha(ip,\boldsymbol{\delta})} e^{-2px+8p^3 t} \chi_-(-ip,x,t), \end{equation}
for $p \in \Gamma$.
\item For fixed $x$ and $t$, $\chi(k,x,t)$ has asymptotic behavior
\begin{equation} \boldsymbol{\chi}(k,x,t) = \mathbf{1} + O(k^{-1}) \text{ as } k \to \infty. \end{equation}
where $\mathbf{1} = (1,1)$.

\end{enumerate}
\end{rhp}

Property (3) of the above nonlocal scalar Riemann--Hilbert problem did not appear in the papers \cite{DZZ16,ZDZ16,ZZD16}, and this oversight was addressed in \cite{NZZ18}.

\begin{proof}
To pin down the function $\xi$ we must additionally set the branch cut to be $i\Gamma$, and take the choice of branch with asymptotic behavior $\xi(k) = 1+O(k^{-1})$ as $k \to \infty$.
The fact that $\psi$ and $\psi_{aux}$ are meromorphic for $\mathbb{C} \setminus i\Gamma$ is inherited from the fact that they are meromorphic functions on $\Sigma$.
Moreover, the fact that product $\psi_{aux}(k,x,t) \psi(k,x,t)$ only has poles on $\left< i \kappa_{2 \ell-1} \right>$ means that $\psi_{aux}(k,x,t) \psi(k,x,t)$ is holomorphic in $\mathbb{C} \setminus i \Gamma$.
These functions are not holomorphic on $i \Gamma$ because the coordinate $k$ for $\Sigma$ does not extend to $i \Gamma$.
By combining these considerations with the fact that the exponential terms are entire, property (1) becomes clear.

From the constructions of $\psi$ and $\psi_{aux}$ as meromorphic functions on $\Sigma \setminus \infty$ and the placement of their poles, it clear that they analytically continue across the jump on two sheets into open regions containing $i\Gamma \setminus \{\text{endpoints}\}$.
Continuity of the boundary values of $\xi$ and $e^{i\alpha(k,x,t)+ikx+4 i k^3 t}$ in $i \Gamma \setminus \{ \text{endpoints} \}$ is also clear.
Therefore, $\chi$ itself has continuous boundary values $\chi_\pm(k,x,t)$ for $k \in i \Gamma \setminus \{\text{endpoints} \}$, proving property (2).

The pole conditions of $\psi(k,x,t) \psi_{aux}(k,x,t)$ at $\left<i \kappa_{2 \ell - 1} \right>$ on the endpoints of the cuts imply that
\begin{equation} \label{eq:sing1} \psi(k,x,t) \psi_{aux}(k,x,t) = b_{2 \ell - 1}(x,t) (i \kappa_{2 \ell - 1} \mp k)^{-\frac{1}{2}} + 
O(1)\end{equation}
for some $b_{2\ell-1}(x,t)$ as $k \to \pm i\kappa_{2 \ell - 1}$, $\ell = 1,2,\dots,g$.

Regularity conditions at $\left<\kappa_{2 \ell}\right>$ imply that
\begin{equation} \label{eq:sing2} \psi(k,x,t) \psi_{aux}(k,x,t) = b_{2 \ell}(x,t) + O((i \kappa_{2 \ell - 1} \mp k)^{\frac{1}{2}})\end{equation}
for some $b_{2 \ell}(x,t)$ as $k \to \pm i\kappa_{2 \ell}$, $\ell = 1,2,\dots,g$.
In the above, the branch cuts of the square roots are chosen to align locally with $i\Gamma$.
The function $e^{i \alpha(k,\boldsymbol{\delta}) + i k x + 4 i k^3 t}$ is entire, so multiplying by it has no effect on the order of the singular behaviors \eqref{eq:sing1}, \eqref{eq:sing2}. 
Multiplying the singular behaviors \eqref{eq:sing1}, \eqref{eq:sing2} by the singular/zero behavior of $\xi(k)$ near $\pm i \kappa_j$ give singular behaviors
\begin{equation} \xi(k) \psi(k,x,t) \psi_{aux}(k,x,t) = \tilde b_j(x,t) (i \kappa_j \mp k)^{-\frac{1}{4}} + O(1) \label{eq:qroot} \end{equation} for some $\tilde b_j(x,t)$ as $k \to \pm i \kappa_j$, $j = 1,2,\dots,2g$.
The brach cuts of the quartic roots are chosen to align locally with the cuts on $i \Gamma$.
Property (3) follows easily from the singular behavior \eqref{eq:qroot} at $\pm i \kappa_j$ because the singular behavior occurs only at a finite number of points.

As a meromorphic function on $\Sigma$, the Baker--Akheizer function $\psi$ satisfies
\begin{equation} \psi_+(ip,x,t) = \psi_+(-ip,x,t), \quad \psi_-(ip,x,t) = \psi_-(-ip,x,t) \label{eq:BAjump}\end{equation}
for $p \in \Gamma$.
The auxiliary Baker--Akheizer function $\psi_{aux}$ also satisfies the jump relation \eqref{eq:BAjump}. 
The function $\xi(k)$ satisfies the jump relation
\begin{equation} \xi_+(ip) = i \text{sgn}(p) \xi_+(-ip), \quad \xi_-(ip) = -i \text{sgn}(p) \xi_-(-ip) \end{equation}
for $p \in \Gamma$.
The function $e^{i \alpha(k,\boldsymbol{\delta})+ikx}$ satisfies the relation
\begin{equation} (e^{i \alpha(k,\boldsymbol{\delta}) + i k x + 4ik^3t})|_{k = ip} = e^{2i\alpha(ip,\boldsymbol{\delta})-2px+8p^3t} (e^{i \alpha(k,\boldsymbol{\delta}) + i k x + 4 i k^3 t})|_{k = -ip}. \end{equation}
Combining these jump relations gives the jump relations appearing in property (4).

Property (5) follows from the fact that $\xi(k)$, $e^{ikx+4ik^3t}\psi(k,x,t)$ and $e^{i\alpha(k,\boldsymbol{\delta})} \psi_{aux}(k,x,t)$ all have asymptotic behaviors $1+O(k^{-1})$ as $k \to \infty$ (for fixed $x,t$).
\end{proof}

\begin{remark}
The nonlocal Riemann--Hilbert problem solved by $\chi$ is equivalent to a local vector Riemann--Hilbert problem solved by $\boldsymbol{\chi} = [\chi(k,x,t),\chi(-k,x,t)]$.
The minor changes to the analogous local vector Riemann--Hilbert problem discussed in \cite{DZZ16,ZDZ16,ZZD16} that need to be made to accommodate finite gap solutions are clear from the conditions on the nonlocal scalar Riemann--Hilbert problem.
One reason that the power in the bound (3) is $-\tfrac{1}{4}$ is important is because it implies that the non-tangential boundary values $\chi_\pm$ are elements of $L^2(i\Gamma) \equiv L^2(\Gamma)$.
\end{remark}

The explicit form of the local vector Riemann--Hilbert problem discussed above is the following:

\begin{rhp} For all $x,t$ find a $1 \times 2$ vector valued function $\boldsymbol{\chi}(k,x,t)$ such that
\begin{enumerate}
\item $\boldsymbol{\chi}$ is a holomorphic function of $k \in \mathbb{C} \setminus i \Gamma$.
\item The boundary values
\begin{equation}\boldsymbol{\chi}_+(ip,x,t) = \lim_{\epsilon \to 0^+} \boldsymbol{\chi}(ip+\epsilon,x,t), \quad \boldsymbol{\chi}_-(ip,x,t) = \lim_{\epsilon \to 0^+} \boldsymbol{\chi}(ip-\epsilon,x,t)\end{equation}
of $\boldsymbol{\chi}$ for $p \in  \Gamma \setminus\{\text{endpoints of }\Gamma\}$ are continuous.
\item  For fixed $x,t$, there exists a constant $C(x,t)$ such that
\begin{equation} \chi_1(k,x,t),\chi_2(k,x,t) \le C(x,t) |k \pm i\kappa_j|^{-1/4} \end{equation}
for $k$ in a neighborhood of $\pm i \kappa_j$.
\item The boundary values $\boldsymbol{\chi}_\pm(ip;x,t)$ of $\boldsymbol{\chi}(k,x,t)$ for $p \in \Gamma$ are related by
\begin{equation}
\boldsymbol{\chi}_+(ip,x,t) = \boldsymbol{\chi}_-(ip,x,t) V(p;x,t)
\end{equation}
where
\begin{equation} V(p;x,t) =  \begin{pmatrix}  \frac{1 - R_1(p) R_2(p)}{1 + R_1(p) R_2(p)} & \frac{2 i  R_1(p)}{1 + R_1(p) R_2(p)} e^{-2px+8p^3t} \\ \frac{2 i  R_2(p)}{1 + R_1(p) R_2(s)} e^{2px-8p^3t} & \frac{1 - R_1(p) R_2(p)}{1 + R_1(p) R_2(p)}  \end{pmatrix}. \end{equation}
\item The function $\boldsymbol{\chi}$ has the limiting behaviors $\chi(k) \to 1$ as $k \to \infty$.
\item $\boldsymbol{\chi}$ satisfies the symmetry
\begin{equation} \boldsymbol{\chi}(-k,x,t) = \boldsymbol{\chi}(k,x,t) \begin{pmatrix} 0 & 1 \\ 1 & 0 \end{pmatrix}. \end{equation}
\end{enumerate}
\end{rhp}

\begin{proof}[Proof of Theorem \ref{thm:fingap}]
We make the assumption that the solution to the Riemann--Hilbert problem from proposition \ref{prop:RHP} has the form
\begin{equation} \chi(k,x,t) = 1 + \frac{i}{\pi} \int_\Gamma \frac{\tilde f(s,x,t)}{k- is} ds.\end{equation}
The boundary values of $\chi(i p,x)$ for $p \in \Gamma$ are then given in terms of $\tilde f$ as
\begin{equation} \chi_+(ip,x,t) = 1 + \hat H_{\Gamma} \tilde f(p,x,t) + i \tilde f(p,x,t), \end{equation}
\begin{equation} \chi_-(ip,x,t) = 1 + \hat H_{\Gamma} \tilde f(p,x,t) - i \tilde f(p,x,t), \end{equation}
where $\hat H_{\Gamma}$ is the Hilbert transform with support on $\Gamma$
\begin{equation} \hat H_{\Gamma} \tilde f(k,x,t) = \frac{1}{\pi} \fint_\Gamma \frac{\tilde f(s,x,t)}{k-s}ds.  \end{equation}
The jump conditions on $\chi$ give the system of integral equations
\begin{equation} 1 + \hat H_{\Gamma} \tilde f(p,x,t) + i \tilde f(p,x,t) = i \text{sgn}(p) e^{2i\alpha(ip,\boldsymbol{\delta})} e^{-2px+8p^3t} (1 + \hat H_{\Gamma} \tilde f(-p,x,t) + i \tilde f(-p,x,t)),  \end{equation}
\begin{equation} 1 + \hat H_{\Gamma} \tilde f(p,x,t) - i \tilde f(p,x,t) = -i \text{sgn}(p) e^{2i\alpha(ip,\boldsymbol{\delta})} e^{-2px+8p^3t} (1 + \hat H_{\Gamma} \tilde f(-p,x,t) - i \tilde f(-p,x,t)).  \end{equation}
These are equivalent to
\begin{equation} \tilde f(p,x,t) - \text{sgn}(p) e^{2 i\alpha(ip,\boldsymbol{\delta})} e^{-2px+8p^3t} \hat H_{\Gamma} \tilde f(-p,x,t) = \text{sgn}(p) e^{2 i\alpha(ip,\boldsymbol{\delta})} e^{-2px+8p^3t},  \end{equation}
\begin{equation} \tilde f(-p,x,t) + \text{sgn}(p) e^{-2 i\alpha(ip,\boldsymbol{\delta})} e^{2px-8p^3t} \hat H_{\Gamma} \tilde f(p,x,t) = -\text{sgn}(p) e^{-2 i\alpha(ip,\boldsymbol{\delta})} e^{2px-8p^3t}.  \end{equation}

Define two functions $f,g : [k_1,k_2] \to \mathbb{R}$ by
\begin{equation} f(p,x,t) = \begin{cases} \tilde f(p,x,t) & p \in \Gamma \cap [k_1,k_2] \\
0 & otherwise \end{cases}, \end{equation}
\begin{equation} g(p,x,t) = \begin{cases} -\tilde f(-p,x,t) & p \in \Gamma \cap [k_1,k_2] \\
0 & otherwise \end{cases}. \end{equation}
Then $f$ and $g$ solve
\begin{align}
f(p,x,t)+\frac{R_1(p)}{\pi}e^{-2px+8p^3t}\left[\int_{k_1}^{k_2} \frac{f(q,x,t)}{p+q}dq+\fint_{k_1}^{k_2}\frac{g(q,x,t)}{p-q}dq\right]=R_1(p)e^{-2px+8p^3t},  \\
g(p,x,t)+\frac{R_2(p)}{\pi}e^{2px-8p^3t}\left[\fint_{k_1}^{k_2} \frac{f(q,x,t)}{p-q}dq+\int_{k_1}^{k_2}\frac{g(q,x,t)}{p+q}dq\right]=-R_2(p)e^{2px-8p^3t},
\end{align}
where
\begin{align} & R_1(p) = \exp\left( \sum_{j=1}^g a_j p^{2j-1} \right) \sum_{\ell=1}^g \mathbbm{1}_{[\kappa_{2\ell-1}, \kappa_{2 \ell}]} (p),  \\
& R_2(p) = \exp\left( -\sum_{j=1}^g a_j p^{2j-1} \right) \sum_{\ell=1}^g \mathbbm{1}_{[\kappa_{2\ell-1}, \kappa_{2 \ell}]} (p),   \end{align}
are defined in terms of the real coefficients $a_j = (-1)^j 2 t_j(\boldsymbol{\delta})$.

\end{proof}

\section{Numerical Primitive Solutions}

The primitive potential/solution method for the KdV equation lends itself to numerical evaluations because of the simplicity of the system of singular integral equations (\ref{eq:f},\ref{eq:g}) \cite{DZZ16,ZDZ16,ZZD16,NZZ18}.
The only difficulty is that the matrices discretizing the system (\ref{eq:f},\ref{eq:g}) are badly conditioned, and so the matrix inversions must be computed with arbitrary precision arithmetic.
A regularization method that allows system (\ref{eq:f},\ref{eq:g}) to be solved with double precision arithmetic would be invaluable.

In this section we will numerically compute finite gap solutions as primitive potentials in the $g=2$ and $g=3$ gap cases.
For the $g=2$ gap case, we also consider the following cases numerically:
\begin{itemize}
\item The case where
\begin{equation} R_1(p)=\exp\left( \sum_{j = 1}^g a_j p^{2j-1} \right) \sum_{\ell = 1}^g \mathbbm{1}_{[\kappa_{2\ell-1}, \kappa_{2 \ell}]}(p), \quad R_2=0. \end{equation}
These are steplike solutions that approach finite gap solutions as $x \to -\infty$ and approach $0$ as $x \to \infty$.
These evolve into dispersive shockwave type solutions.
\item The case where
\begin{equation} R_1(p)= R \exp\left( \sum_{j = 1}^g a_j p^{2j-1} \right) \sum_{\ell = 1}^g \mathbbm{1}_{[\kappa_{2\ell-1}, \kappa_{2 \ell}]}(p), \end{equation} \begin{equation} R_2(p)= R \exp\left(- \sum_{j = 1}^g a_j p^{2j-1} \right) \sum_{\ell = 1}^g \mathbbm{1}_{[\kappa_{2\ell-1}, \kappa_{2 \ell}]}(p) \end{equation}
for $R > 0$.
For large $R$, these solutions appear to have a high amplitude region and a phase modulation that is localized in space-time.
\end{itemize}

\begin{figure}
\includegraphics[width=\textwidth]{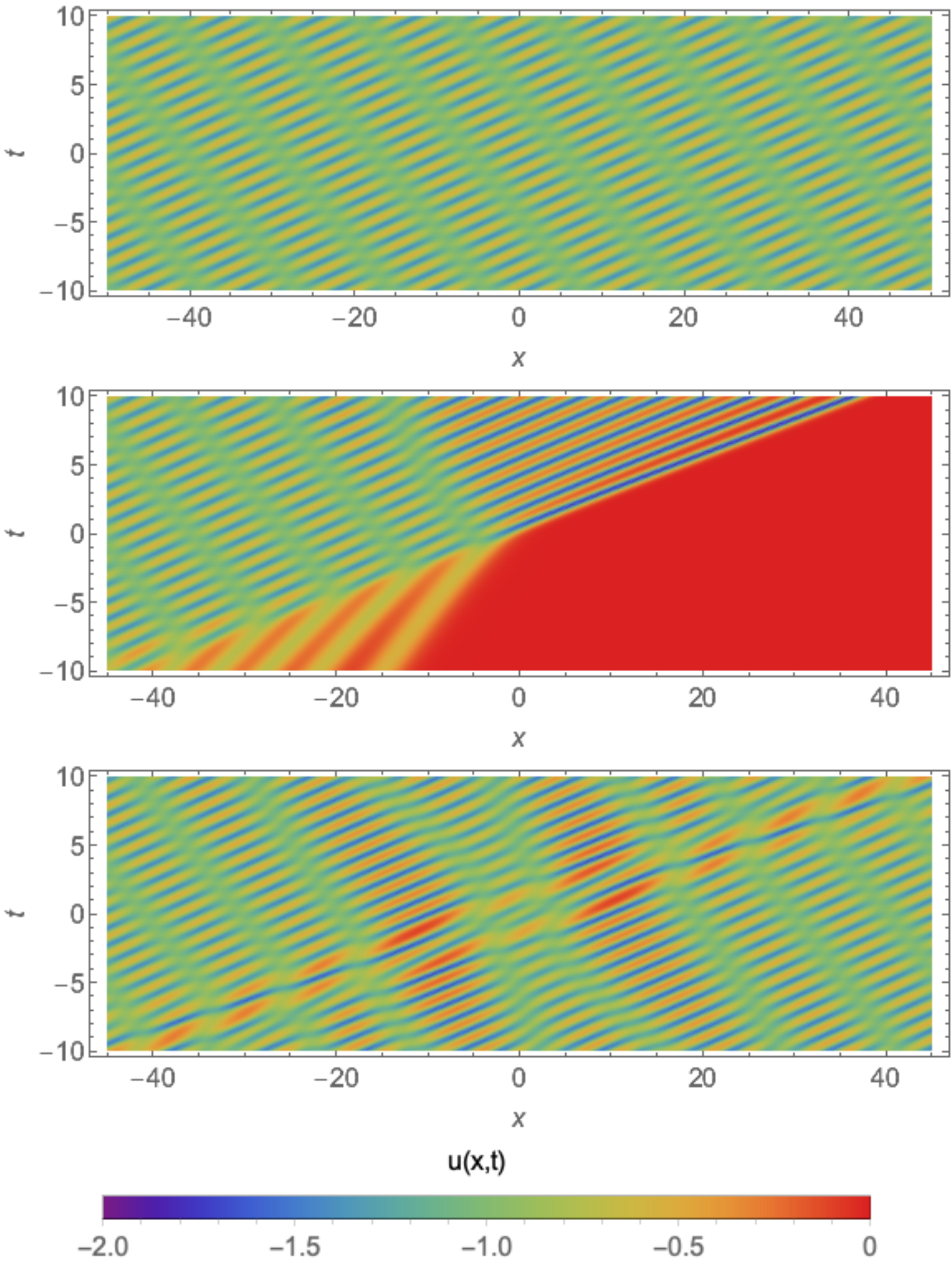}
\caption{Comparison of primitive solutions corresponding to an algebro-geometric two gap solution, a one sided two gap solution, and a two sided two gap solution with $R=100$.
All three correspond to $\kappa_1=\frac{1}{2}, \kappa_2=\frac{\sqrt{2}}{2}, \kappa_3 = \frac{\sqrt{3}}{2}, \kappa_4=1$ and $a_1=a_2=0$.}
\end{figure}

\begin{figure}
\includegraphics[width=\textwidth]{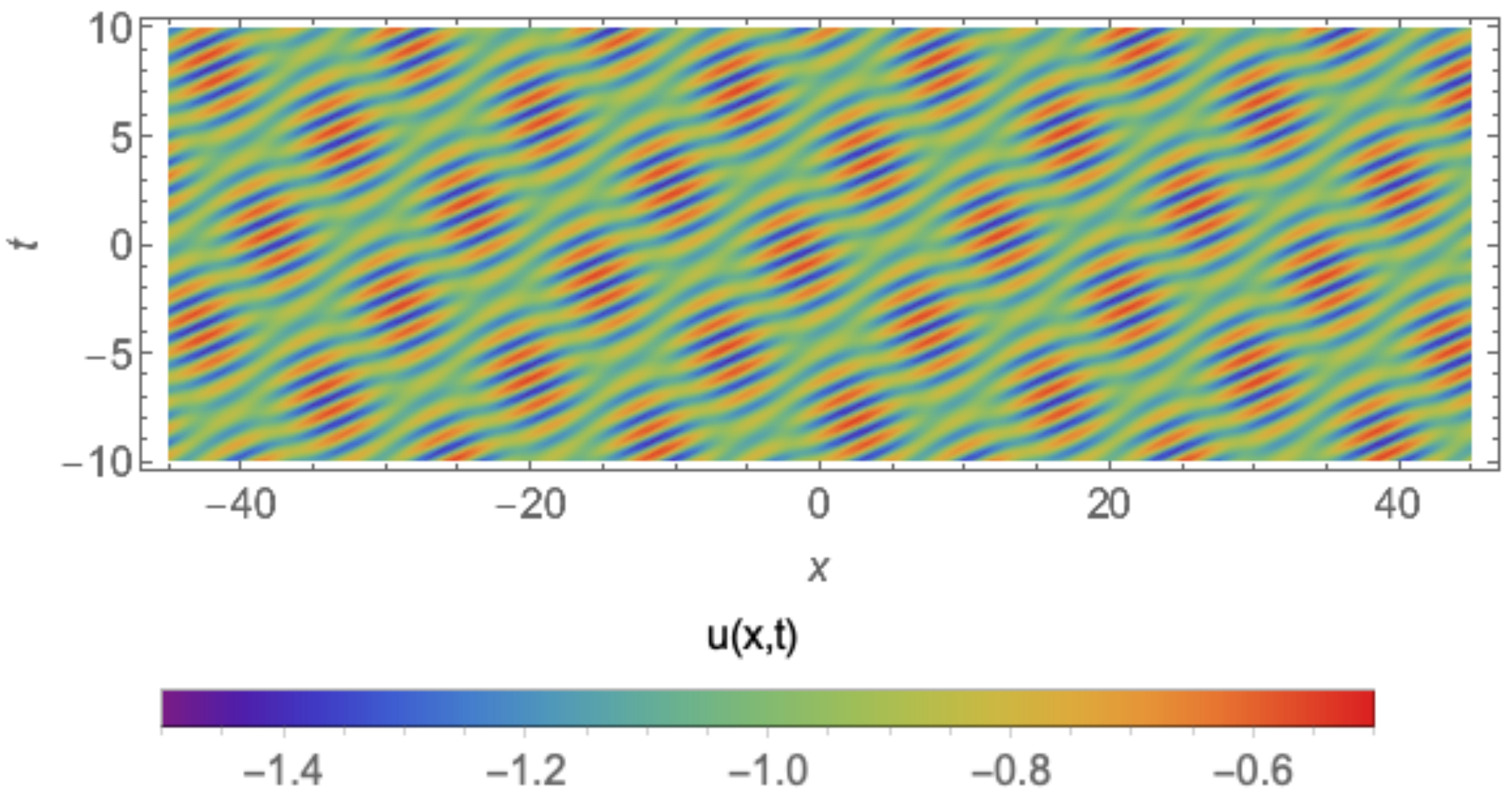}
\caption{A primitive solution corresponding to an algebro-geometric three gap potentials correspond to $\kappa_1=\frac{1}{\sqrt{6}}, \kappa_2=\frac{1}{\sqrt{3}},\kappa_3=\frac{1}{\sqrt{2}},\kappa_4=\frac{\sqrt{2}}{\sqrt{3}},\kappa_5 = \frac{\sqrt{5}}{\sqrt{6}}$, $\kappa_6=1$, $a_1=a_2=0$ and $a_3=6$.}
\end{figure}

\begin{figure}
\includegraphics[width=\textwidth]{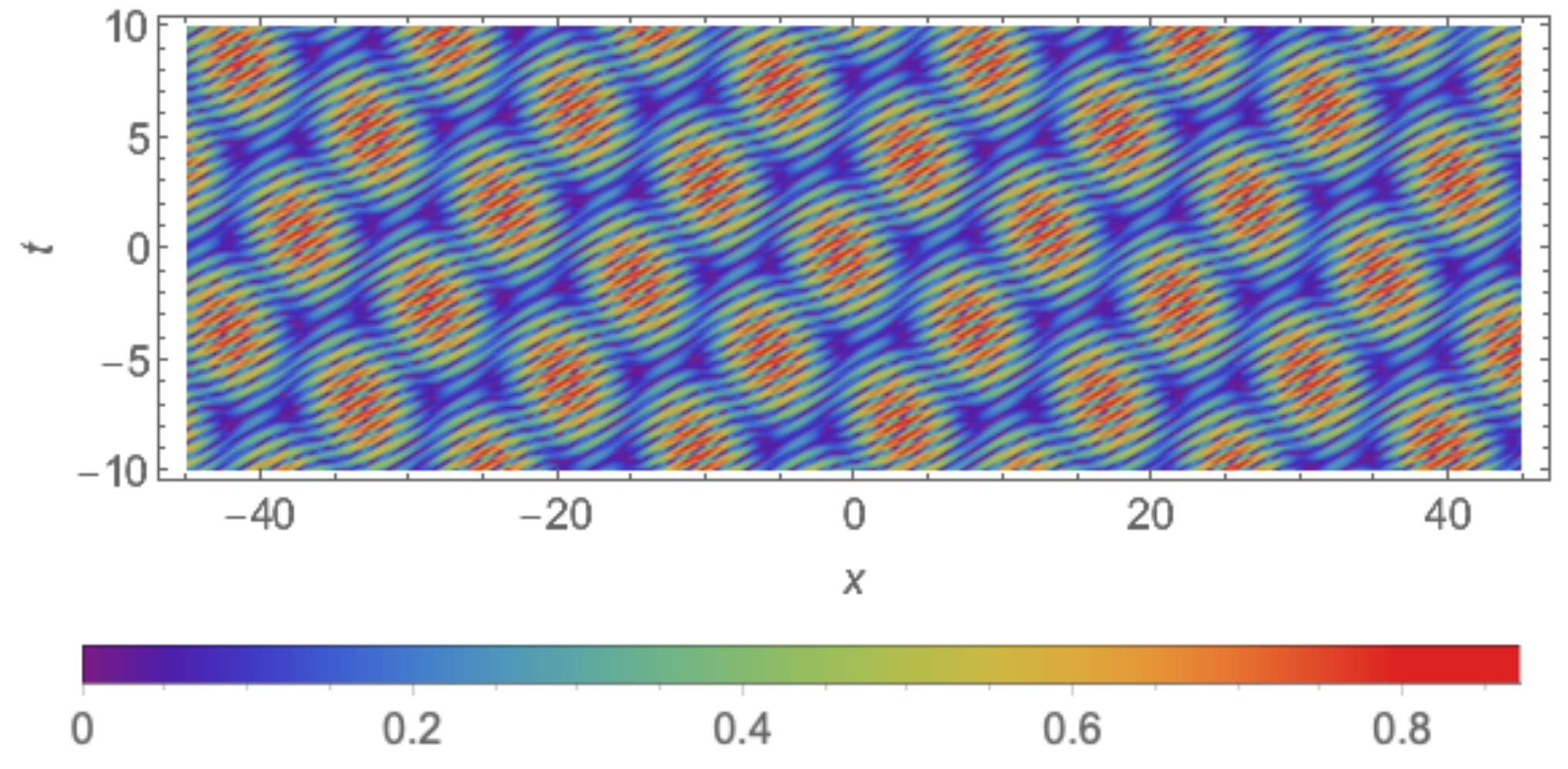}

\caption{The absolute difference between the three gap solutions to the KdV equation appearing in figure 2 and the three gap solution determined by all the same parameters except $a_3=0$.}
\end{figure}

We compute all these solutions numerically as follows:
\begin{enumerate}
\item We discretize the integral equations (\ref{eq:f},\ref{eq:g}) using Gauss--Legendre quadrature on each component of the support of $R_1$ and $R_2$, and handle the singularity be ignoring it (i.e. we set the singular matrix elements to 0).
\item We use the same Gauss--Legendre quadrature rule to evaluate the integral appearing in formula (\ref{eq:u}).
\item We compute the derivative appearing in (\ref{eq:u}) on a uniform spatial grid spectrally using a fast Fourier transform, and we use a Butterworth filter to remove oscillatory artifacts due Gibbs phenomena on the edges of the spatial grid.
\end{enumerate}
An alternate method to evaluating (\ref{eq:u}) that could be implemented at an arbitrary set of space time points $(x,t)$ would be to compute $f_x(p,x,t)$ and $g_x(p,x,t)$ after computing $f(p,x,t)$ and $g(p,x,t)$ by solving the following system of singular integral equations
\begin{align}
& f_x(p,x,t)+\frac{R_1(p)}{\pi}e^{-2px+8p^3t}\left[\int_{k_1}^{k_2} \frac{f_x(q,x,t)}{p+q}dq+\fint_{k_1}^{k_2}\frac{g_x(q,x,t)}{p-q}dq\right] \\
\nonumber & \quad \quad = -2 p  R_1(p)e^{-2px+8p^3t} \left( 1 -  \int_{k_1}^{k_2} \frac{f(q,x,t)}{p+q}dq - \fint_{k_1}^{k_2}\frac{g(q,x,t)}{p-q}dq \right) ,  \\
&g_x(p,x,t) + \frac{R_2(p)}{\pi}e^{2px-8p^3t} \left[\fint_{k_1}^{k_2} \frac{f_x(q,x,t)}{p-q}dq+\int_{k_1}^{k_2}\frac{g_x(q,x,t)}{p+q}dq \right] \\
\nonumber &\quad \quad = -2 p R_2(p)e^{2px-8p^3t} \left(1 + \fint_{k_1}^{k_2} \frac{f(q,x,t)}{p-q}dq+\int_{k_1}^{k_2}\frac{g(q,x,t)}{p+q}dq \right),
\end{align}
which can be derived by differentiating system (\ref{eq:f},\ref{eq:g}) in $x$.
We can then compute (\ref{eq:u}) via
\begin{equation} u(x,t) = \frac{2}{\pi} \int_{k_1}^{k_2} f_x(q,x,t) + g_x(q,x,t) dq. \end{equation}
We opted against this method in this paper because this would require a second ill-conditioned matrix inversion and is, therefore, much slower than differentiation via the fast Fourier transform.

Figure 1 shows space-time plots of the $g=2$ gap numerical solutions, figure 2 shows a space-time plot of a periodic $g=3$ gap numerical solution, and figure 3 shows a space-time plot of the absolute difference between two $g=3$ gap numerical solutions. 

\section{Conclusions}

We have shown how algebro-geometric finite gap potentials can be computed using primitive solutions.
This produces an effective way to generate sequences of shifted N-soliton solutions (remark \ref{rmk:shifty}) that converge to any algebro-geometric finite gap solution in any compact region of space-time.
We have also demonstrated numerically that by modifying the dressing functions for the algebro-geometric finite gap solution it is possible to compute interesting potentials that have finitely many spectral gaps but are not algebro-geometric finite gap potentials.
In particular, the third plot in figure 1 shows a primitive solution that appears to have a disturbance in space-time that is localized near the origin.
This naturally leads us to the following question:
Is it possible to rigorously describe the behavior of this type of solution to the KdV equation near the origin?
A rigorous asymptotic description of primitive potentials with a single spectral gap and $R_2=0$ via nonlinear steepest descent was given in \cite{GGJM18}.
Therefore, it may be possible to build on the approach of \cite{GGJM18} to answer this question.

%\begin{figure}
%\includegraphics[width=0.75\textwidth]{"KdVgenus2".pdf}
%\caption{A genus 2 solution to the KdV equation based on a numerical approximation to the primitive potential method. This one is determined by $\kappa_1=1/4, \kappa_2 = 1/2, \kappa_3 = 3/4, \kappa_4 = 1$ and $a_1 = 2$, $a_2=-3$.
%The integrals are discretized with 25 point quadrature rules on each component of $p \in [\kappa_1,\kappa_2]\cup[\kappa_3, \kappa_4]$ to solve for both $f$ and $g$.}
%\end{figure}
\clearpage
\section{Acknowledgments}

I would like to thank Dmitry and Vladimir Zakharov for introducing me to primitive potentials. I would like to thank Dmitry Zakharov for bringing the importance of the higher KdV flows to my attention. I would like to thank Sergey Dyachenko for providing some advice on computing primitive potentials numerically.

I would like to thank the reviewers for their insightful suggestions for this paper. The reviewers comments helped improve the paper.

This research was supported in part by the National Science Foundation under grant DMS-1715323.


\begin{thebibliography}{00}

\bibitem{Dubrovin75} Dubrovin, B. A.  ``The inverse scattering problem for periodic finite-zone potentials.'' {\it Funct. Anal. Appl.} 9:337-340, 1975.

\bibitem{Dubrovin76} Dubrovin, B. A.  ``Finite-zone linear operators and Abelian varieties.'' {\it Russian Math. Surveys} 31:259-260, 1976.

\bibitem{DN74} Dubrovin, B. A. and S. P. Novikov. ``Periodic and conditionally periodic analogs of the many-soliton solutions of the Korteweg-de Vries equation.'' {\it Soviet Physics JETP}, 40:1058-1063, 1974.

\bibitem{DMN76} Dubrovin, B. A., V.B. Matveev and S. P. Novikov. ``Nonlinear equations of Korteweg-de Vries type, finite-zone linear operators and Abelian varieties.'' {\it Russian Math. Surveys}, 31:56-134, 1976.

\bibitem{DNZZ20} S. A. Dyachenko, P. Nabelek, D. V. Zakharov, V. E. Zakharov, ``Primitive solutions of the Korteweg--de Vries equation,'' {\it TMF}, 202:3 (2020),  3820-392, {\it Theoret. and Math. Phys.}, 202:3 (2020), 334-343

\bibitem{DZZ16} Dyachenko, S., D. Zakharov and V. Zakharov. ``Primitive potentials and bounded solutions of the KdV equation.'' {\it Phys. D}, 333:148-156, 2016.

\bibitem{GGKM} Gardener, S. S., J. M. Greene, M. D. Kruskal, and R. M. Miura. ``Method for solving the Korteweg--de Vries equation.'' {\it Phys. Rev. Lett.} 19(19):1095-1097, 1967.

\bibitem{GGJM18} Girotti, M., T. Grava, R. Jenkins and K. McLaughlin, ``Rigorous asymptotics of a KdV soliton gas.'' 2018 (arXiv:1807.00608).



\bibitem{GS95} Grinevich, P. G. and M. U. Schmidt, ``Period preserving nonisospectral flows and the moduli space of periodic solutions of soliton equations,'' {\it Phys. D}, 87:73-98, 1996.

\bibitem{HM06} Helfrich, K. R. and W. K. Melville. ``Long nonlinear internal waves,'' {\it Annu. Rev. Fluid Mech.} 38:395-425., Palo Alto: Annual Reviews, 2006.


\bibitem{ItsMatv75a}{Its, A. R. and V. B. Matveev. ``Hill operators with a finite number of lacunae.'' {\it Funct. Anal. Appl.} 9:65-66, 1975.}

\bibitem{ItsMatv75b}{Its, A. R. and V. B. Matveev. ``Schrodinger operators with the finite-band spectrum and the n-soliton solutions of the Korteweg--de Vries equation.'' {\it Theor. Math. Phys.} 23:343-355, 1975.}

\bibitem{ItsMatv76}{Its, A. R. and V. B. Matveev. ``On a class of solutions of the KdV equations.'' {\it Prob. Matem. Phys.} 9:65-66, 1976.}

\bibitem{KdV} Korteweg, D. J. and G. de Vries, "On the Change of Form of Long Waves Advancing in a Rectangular Canal, and on a New Type of Long Stationary Waves," {\it Philosophical Magazine} 39 (240): 422-443, 1895.

\bibitem{Lax} Lax, P. D. ``Integrals of nonlinear equations of evolution and solitary waves,'' {\it Comm. Pure Appl. Math.} 21:467-490, 1968.

\bibitem{Lax75} Lax, P. D. ``Periodic solutions of the KdV equation.'' {\it Comm. Pure Appl. Math.} 28:141-148, 1975.

\bibitem{Mar74} Marchenko, V. A. ``Periodic problem of Korteweg de Vries Equation I,'' {\it Matem. Sbornik}, 95:331-356, 1974 [In Russian]

\bibitem{Mar91} Marchenko, V. ``The Cauchy problem for the KdV equation with non-decreasing initial data,'' in: V.E. Zakharov (Ed.), What is integrability? in: Springer Series in Nonlinear Dynamics, Springer-Verlag, Berlin, 1991.

\bibitem{MaOs75} Marchenko, V. A. and I. V. Ostrovskii, ``Characterization of spectrum of Hill's operators.'' {\it Math. USSR Sb.} 97:540-586, 1975.


\bibitem{Mat75} Matveev, V. B. ``New scheme of integration of the Korteweg de Vries equation (Lecture given at the Petrovskij seminar of 26 March 1975 in Moscow),'' {\it Uspekhi Matem. Nauk} 30:201-203, 1975.

\bibitem{MckTr76} McKean, H. P. and E. P. Trubowitz, ``Hill's operator and hyperelliptic function theory in the presence of infinitely many branch points.'' {\it Comm. Pure Appl. Math.} 29:143-226, 1976.

\bibitem{MckTr78} McKean, H.P. and E. P. Trubowitz. ``Hill's surfaces and their theta functions.'' {\it Bull. Amer. Math. Soc.} 84(6):1042-1085, 1978.

\bibitem{MckMo75}{McKean, H. P. and P. van Moerbeke. ``The spectrum of Hill's equation.'' {\it Invent. Math.}, 30(3):217-274, 1975.}

\bibitem{MN19} K. T-R McLaughlin and P. V. Nabelek (2019) ``A Riemann--Hilbert Problem Approach to Infinite Gap Hill's Operators and the Korteweg--de Vries Equation,''
{\it International Mathematics Research Notices}, Vol. 00, No. 0, pp. 1?65
doi:10.1093/imrn/rnz156.


\bibitem{NZZ18} Nabelek, P., D. Zakharov and V. Zakharov, ``On symmetric primitive potenitals,'' {\it To appear in J. Int. Sys.} 2019 (arXiv:1812.10545).

\bibitem{NZ20} Nabelek, P., Zakharov, V. ``Solutions to the Kaup--€"Broer system and its (2+1) dimensional integrable generalization via the dressing method.'' {\it  Phys. D: Nonlinear Phenomena}. 409:132478", 2020, https://doi.org/10.1016/j.physd.2020.132478.


\bibitem{No74} Novikov, S. P. ``A periodic problem for the Korteweg-de Vries equation,'' {\it I. Funct. Anal. Appl.} 8:236-246, 1974.

\bibitem{ZMNP84} Novikov, S., S. Manakov, L. Pitaevskii and V. Zakharov, {\it Theory of solitons, The inverse scattering method}, Contemporary Soviet Mathematics, 1984.

\bibitem{Osborne2010}{Osborne, A. {\it Nonlinear Ocean Waves and the Inverse Scattering Transform}, Academic Press, New York, New York, 1 edition, 2010.}

\bibitem{TD13} Trogdon, T. and B. Deconinck, ``A Riemann-Hilbert problem for the finite-genus solutions of the KdV equation and its numerical solution,'' {\it Phys. D}, 251:1-18, 2013.

\bibitem{ZDZ16} Zakharov, D., S. Dyachenko and V. Zakharov, ``Bounded solutions of KdV and non-periodic one-gap potentials in quantum mechanics,'' {\it Lett. Math. Phys.} 106(6):731-740, 2016.

\bibitem{ZF71} Zakharov, V. E. and L. D. Faddeev, ``Korteweg-de Vries equation: A completely integrable Hamiltonian system,'' {\it Funkt. Anal. i Prilozhen.} 5(4):18-27, 1971; {\it Funct. Anal. Appl.} 5(4):280-287, 1971.

\bibitem{ZakharovManakov} Zakharov, V. and S. Manakov, ``Construction of higher-dimensional nonlinear integrable systems and their solutions,'' {\it Funct. Anal. Appl.} 19(2):89-101, 1985.

\bibitem{ZZD16}{Zakharov, D., V. Zakharov and S. Dyachenko, ``Non-periodic one-dimensional ideal conductors and integrable turbulence,'' {\it Phys. Lett. A}, 380(46):3881-3885, 2016.}

\end{thebibliography}
\end{document}